\documentclass[reprint,longbibliography,aps,pra]{revtex4-2}
\usepackage[english]{babel}
\addto\captionsenglish{}
\usepackage{amsmath,amsthm,amssymb,bbm,amsfonts,color,graphicx,xcolor,framed}
\usepackage{url}
\usepackage{tikz}
\usepackage{enumerate}
\usepackage{mathrsfs}
\usepackage{relsize}
\usepackage[colorlinks=true, allcolors=blue]{hyperref}
\usepackage[utf8]{inputenc}
\usepackage{listings}
\lstdefinestyle{mystyle}{
    backgroundcolor=\color{backcolour},
    commentstyle=\color{codegreen},
    keywordstyle=\color{magenta},
    numberstyle=\tiny\color{codegray},
    stringstyle=\color{codepurple},
    basicstyle=\ttfamily\footnotesize,
    breakatwhitespace=false,
    breaklines=true,
    captionpos=b,
    keepspaces=true,
    numbers=left,
    numbersep=5pt,
    showspaces=false,
    showstringspaces=false,
    showtabs=false,
    tabsize=2
}
\newcommand{\om}{\omega}

\newcommand{\ka}{\kappa}

\newcommand{\la}{\lambda}

\newcommand{\al}{\alpha}

\newcommand{\ii}{\mathrm{i}}
\newcommand{\ee}{\mathrm{e}}

\newcommand{\eq}[1]{Eq.~(\ref{eq:#1})}
\newcommand{\fig}[1]{Fig.~\ref{fig:#1}}
\renewcommand{\sec}[1]{Sec.~\ref{sec:#1}}

\newcommand{\app}[1]{Appendix~\ref{sec:#1}}

\newcommand{\SWAP}{\mathrm{SWAP}}

\newcommand{\mc}[1]{\mathcal{#1}}
\newcommand{\id}{\mathbbm{1}}
\newcommand{\tr}{\mathrm{tr}}

\newcommand{\nn}{\nonumber}

\newcommand{\ket}[1]{\left\lvert #1 \right\rangle} 
\newcommand{\bra}[1]{\left\langle #1 \right\rvert} 
\newcommand{\ketbra}[2]{\ket{#1}\!\!\bra{#2}}

\let\perptmp\perp
\renewcommand{\perp}{{\! \mathsmaller{\perptmp}}}

\newcommand{\mrm}{\mathrm}

\newcommand{\nodagger}{{\phantom{\dagger}}}

\newcommand{\dd}{\mathrm{d}}

\def\bgamma{{\boldsymbol \gamma}}
\def\bbeta{{\boldsymbol \beta}}

\newtheorem{theorem}{Theorem}[section]

\theoremstyle{definition}
\newtheorem{definition}{Definition}[section]

\theoremstyle{remark}

\newtheorem{lemma}[theorem]{Lemma}

\usepackage{newfloat}
\DeclareFloatingEnvironment[
    name=Algorithm
]{algorithm}

\makeatletter
\def\@bibdataout@aps{%
\immediate\write\@bibdataout{%
@CONTROL{%
apsrev41Control%
\longbibliography@sw{%
    ,author="08",editor="1",pages="1",title="0",year="1"%
    }{%
    ,author="08",editor="1",pages="1",title="",year="1"%
    }%
  }%
}%
\if@filesw \immediate \write \@auxout {\string \citation {apsrev41Control}}\fi
}
\makeatother

\begin{document}

\author{Joris Kattem\"olle}
\author{Guido Burkard}
\affiliation{Department of Physics, University of Konstanz, D-78457 Konstanz, Germany}
\title{Ability of error correlations to improve the performance\\ of variational quantum algorithms}

\begin{abstract}
  The quantum approximate optimization algorithm (QAOA) has the potential of providing a useful quantum advantage on noisy intermediate-scale quantum (NISQ) devices. The effects of uncorrelated noise on variational quantum algorithms such as QAOA have been studied intensively. Recent experimental results, however, show that the errors impacting NISQ devices are significantly correlated. We introduce a model for both spatially and temporally (non-Markovian) correlated errors based on classical environmental fluctuators. The model allows for the independent variation of the marginalized spacetime-local error probability and the correlation strength. Using this model, we study the effects of correlated stochastic noise on QAOA. We find evidence that the performance of QAOA improves as the correlation time or correlation length of the noise is increased at fixed local error probabilities. This shows that noise correlations in themselves need not be detrimental for NISQ algorithms such as QAOA.
\end{abstract}

\maketitle

\section{Introduction}\label{sec:introduction}
Quantum computers hold the promise of outperforming classical computers on tasks such as the simulation of quantum-mechanical systems~\cite{feynman1982simulating,lloyd1996universal}, factoring~\cite{shoralgorithms}, unstructured database search~\cite{grover1996fast}, and solving linear systems of equations~\cite{harrow2009quantum}. However, it is much harder to shield qubits, the fundamental building blocks of quantum computers, from environmental noise than it is to shield classical bits~\cite{suter2016protecting}, leading to errors in quantum computations. These errors can be detected and corrected, provided that error probabilities remain below a certain threshold, as is proven by threshold theorems for fault-tolerant quantum computation~\cite{shor1996fault,aharonov1997fault, knill1998resilient,aharonov2008fault}. Fault-tolerant quantum computation remains possible in the presence of some forms of spatially and temporally correlated errors~\cite{terhal2005fault, aharonov2006fault, aliferis2006quantum, aharonov2008fault,preskill2013sufficient} but breaks down in others~\cite{clader2021impact}. In any case, the quantum overhead required for fault-tolerant quantum computation is currently prohibitively large.

Nevertheless, current pre-error-corrected noisy intermediate-scale quantum (NISQ)~\cite{preskill2018nisq} computers can already outperform classical computers on some tasks and have hence shown a quantum advantage~\cite{arute2019quantum,zhong2020quantum}. However, the tasks for which a quantum advantage has been demonstrated are artificial and have no known applications. Thus, the next milestone in the field will be the demonstration of a \emph{useful} quantum advantage.

Hybrid quantum-classical variational quantum algorithms (VQAs) have the potential of showing a useful quantum advantage on NISQ devices~\cite{cerezo2021variational}. In these algorithms, a parametrized quantum state, known as the ansatz, is prepared using a parameterized quantum circuit.  Subsequently, the expectation value of some observable (which depends on the specific VQA) is estimated by repeated preparation and measurement of the ansatz. The parameters and the resulting expectation value are fed into a classical optimization algorithm, which in turn suggests new parameters. By repeating the above steps, the optimizer heuristically seeks parameters that optimize the expectation value of the observable. The VQA returns the optimal value of the observable and the parameters at which it was attained. Using these parameters, the optimal ansatz state can be reconstructed, which allows for the extraction of additional classical data from the optimal ansatz state.

\begin{figure}[b]
  \begin{center}
\vspace{-1em}
    \includegraphics[scale=1.1]{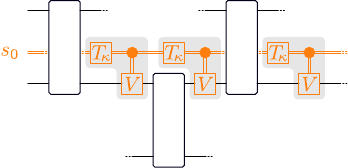}
    \vspace{-1em}
  \end{center}
  \caption{\label{fig:error_model} The temporal fluctuator model. Each qubit interacts with an independent classical two-level fluctuator in the bath (orange, only one fluctuator shown). The fluctuator is initially excited with probability $p$, as described by the classical ensemble $s_0$. The operation $T_{\kappa}$ resets the fluctuator to $s_0$ with probability $1-\kappa$. We refer to $\ka$ as the correlation strength. The unitary error operator $V$ is applied to the qubit if and only if the fluctuator is excited, leading to an error model with independent time-local error probability $p$ and correlation time $0\leq\tau=-1/\ln{\kappa}<\infty$ (in units of gate time). For a circuit of depth $m$, acting on $n$ qubits, the full temporal fluctuator model is obtained by repeating the shaded area $m$ times and the entire structure $n$ times. Empty rectangles represent generic gates. At the end of the circuit, the fluctuator is discarded.}
\end{figure}

The quantum approximate optimization algorithm (QAOA)~\cite{farhi2014quantum} is a VQA designed to (approximately) solve instances from a large set of optimization problems~\cite{lucas2014ising}, including problems of practical relevance, such as portfolio optimization~\cite{hodson2019portfolio} and correlation clustering~\cite{weggemans2022solving}. The cost function of the optimization problem is first mapped to an Ising-type Hamiltonian whose ground state corresponds to the optimal solution of the optimization problem. This Hamiltonian is the observable being measured during the optimization loop. In the context of QAOA, the expectation value of the Hamiltonian, as a function of the variational parameters, is called the cost function landscape. QAOA seeks approximate ground states of the Hamiltonian by optimizing the cost function landscape and thereby finds approximate solutions to the optimization problem.

A common figure of merit for QAOA is the approximation ratio (AR), defined as the ratio of the cost function landscape at the best parameters found by QAOA to the global optimal value of the classical cost function.  Along the lines of Ref.~\cite{fontana2021evaluating}, to evaluate the effects of noise, a \emph{noise-unaware} AR can be defined. The noise-unaware AR is obtained by first running QAOA in classical simulations with all noise turned off. The resulting optimal parameters are subsequently used to compute the expectation value of the Hamiltonian while the noise is turned on. An AR that is larger than the noise-unaware AR indicates that variational parameters can adapt to the presence of noise. We call this ability the \emph{noise adaptivity} of the optimal parameters.

The behavior of VQAs such as QAOA under \emph{local} errors has been studied extensively~\cite{xue2021effects, marshall2020characterizing, mcclean2016theory, malley2016scalable, mcclean2016theory, malley2016scalable, mcclean2017hybrid,colless2018computation, gentini2020noise,sharma2020noise,fontana2021evaluating}, showing that VQAs possess forms of resilience against coherent~\cite{mcclean2016theory, malley2016scalable} and incoherent~\cite{mcclean2017hybrid,colless2018computation, gentini2020noise,sharma2020noise,fontana2021evaluating} errors. VQAs~\cite{mcclean2017hybrid}, specifically QAOA~\cite{farhi2016quantum}, are therefore thought to be ideal candidates for obtaining a useful quantum advantage.

Recent experimental results have highlighted the ubiquity of both spatially and temporally correlated errors in NISQ devices. Sources of spatially correlated errors include crosstalk~\cite{heinz2021crosstalk,berg2022probabilistic}, fluctuations of external, quasihomogeneous magnetic fields~\cite{monz2011qubit}, and (cosmic) radiation~\cite{wilen2021correlated,cardani2021reducing}. Sources of temporally correlated errors include $1/f$ noise (where $f$ denotes frequency)~\cite{dutta1981low,koch1983flicker, koch2007model,connors2022charge}, reflections in drive lines~\cite{reed2013entanglement}, long-lived quasiparticle excitations~\cite{wilen2021correlated}, nuclear spins~\cite{Khaetskii2002}, and microscopic two-level systems~\cite{burkard2009non-markovian,schloer2019correlating,burnett2019decoherence}. In particular, Refs.~\cite{schloer2019correlating,burnett2019decoherence} studied the stability of superconducting (transmon) qubits. Telegraph-like switching of relaxation times~\cite{burnett2019decoherence} and qubit frequency~\cite{schloer2019correlating} were observed. Both works attribute these effects to single two-level systems local to each qubit.

Given the importance of VQAs, their resilience against local errors, and the prevalence of correlated errors, an understanding of the effects of correlated errors on VQAs is needed. We make a first step in this direction by studying the effects of temporally correlated (non-Markovian) and spatially correlated errors on QAOA.

We introduce a toy model for temporally correlated errors (\fig{error_model}) that is inspired by the two-level systems local to each qubit of Refs.~\cite{schloer2019correlating,burnett2019decoherence}. Our model captures the essence of correlated errors, which we consider to be the existence of a nontrivial error probability and correlation strength. To this end, we assign a classical two-level fluctuator to each qubit that causes a unitary error operation $V$ on the associated qubit conditioned on the state of the fluctuator. The error probability $p$ and correlation time $\tau$ of these errors, defined in \fig{error_model}, are determined by the internal evolution of the fluctuators. A significant advantage of this model is its ability to fully and independently control the marginalized, time-local error probability $p$ and correlation time $\tau$. Interchanging the roles of space and time, a model for spatially correlated errors is obtained with local error probability $p$ and correlation length $\la$ (\fig{error_model_spatial}). Both $\tau$ and $\la$ arise from the same fluctuator parameter $\ka$ that we call the \emph{correlation strength}. Thus, we are able to treat temporally and spatially correlated errors on an equal footing. Limitations of our model are discussed in \sec{discussion} of this paper.

\begin{figure}
  \begin{center}
    \includegraphics[scale=1.1]{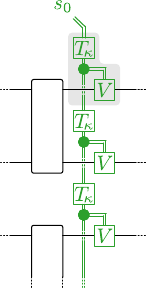}
  \end{center}
  \caption{\label{fig:error_model_spatial} The spatial fluctuator model. Definitions are as in \fig{error_model}. As opposed to  \fig{error_model}, the fluctuator is displayed in green and travels through space rather than time. The error model for one circuit layer is generated by repeating the shaded area in the spatial direction $n$ times, with $n$ the number of qubits. The full error model is obtained by repeating the obtained structure after each gate time. In the spatial fluctuator model, the correlation strength $\ka$ leads to an error correlation length $0\leq\la=-1/\ln{\kappa}<\infty$ in units of interqubit distance.}
\end{figure}

We classically simulate the performance of QAOA under the fluctuator error models of Figs.~\ref{fig:error_model} and~\ref{fig:error_model_spatial} for various problem instances, error probabilities $p$, and correlation strengths $\ka$. Our main results are the following.

Our first result is the observation of an \emph{increase} of the performance of QAOA (as measured by the AR) as a function of the correlation strength $\ka$ at fixed spacetime-local marginalized error probabilities $p$ across all cases studied. We analytically derive the effect of our fluctuator models on the AR to first order in the spacetime-local error probability $p$. Using this linear order theory, we show the increase in AR with correlation strength is explained by a counting argument, stating that a correlated error can happen in fewer ways than an uncorrelated error. Although a single, strongly correlated error can have a stronger effect on the output state than a single error happening at a single spacetime location, in all cases we studied this does not outweigh the fact that there are simply fewer ways in which a correlated error may happen.

A second, distinct result is the observation of a linear degradation of the performance of QAOA as a function of the error probability $p$ around $p=0$, noise adaptivity of QAOA, and the existence of critical error probabilities for this noise adaptivity, all for uncorrelated and correlated errors that do not break the symmetries of the cost function landscape. As $p$ is increased away from $p=0$, the AR and the noise-unaware AR decrease linearly and are essentially equal up until some critical $p$ that depends on the fluctuator model (temporal or spatial), the correlation strength, and problem instance. After this critical $p$, the AR becomes better than the noise-unaware AR. This divergence between the AR and noise-unaware AR coincides with an abrupt jump in optimal parameters that are otherwise essentially constant. This indicates that there are two competing, separated points in the cost function landscape and that one of them overtakes the other at the critical $p$.

Recently, similar linear decay, noise adaptivity of optimal parameters, and critical error probabilities were observed for a VQA by Fontana et al.~\cite{fontana2021evaluating}. Our results regarding these points go beyond the results of Fontana et al., first because we show them for QAOA, and second because our results also hold for correlated errors. Finally, Fontana et al. attributed the noise adaptivity to the fact that their noise model breaks the symmetries of the cost function landscape. Our fluctuator error models do not break the symmetries of the cost function landscape, showing the existence of additional mechanisms for noise adaptivity.

The remainder of this paper is organized as follows. We give a more detailed introduction to QAOA in \sec{QAOA}. In \sec{error_model}, we derive physical properties of our error model, such as correlation functions and marginalized error probabilities. In the same section, we derive the linear order effects of our fluctuator models on the cost function landscape. Subsequently, in \sec{numerical_methods_and_results}, we describe our numerical methods and present and analyze our  results, followed by a discussion in \sec{discussion}.

\section{QAOA}\label{sec:QAOA}
A wide class of optimization problems can be formulated using a quadratic cost function, $C(z)=\sum_{i<j}^{n}\om_{ij}z_iz_j+\sum_{i=1}^{n} \om_i z_i$, with $z_i\in\{-1,1\}$~\cite{lucas2014ising}. The goal is to find a $z=(z_1,\ldots,z_n)$ that optimizes $C$ globally. Depending on the specific application, the optimization is a maximization or minimization of $C$. The cost function can readily be mapped to a Hamiltonian on $n$ qubits,
\begin{equation}\label{eq:Hamiltonian}
  H=\sum_{i<j}^{n}\om_{ij}Z_iZ_j+\sum_{i=1}^n \om_i Z_i,
\end{equation}
with $Z_i$ the Pauli-$Z$ operator acting on qubit $i$. The mapping is such that $C(z)=\bra z H \ket z$. The \emph{Sherrington–Kirkpatrick} (SK) model encompasses all
cost functions with $\om_{ij}\in\{1,-1\}$ and $\om_i=0$. In this case, both $C$ and $H$ can be identified with the same undirected weighted graph with $n$ nodes and adjacency matrix $\om_{ij}$. We do not assume the SK model in this section unless stated otherwise.

QAOA~\cite{farhi2014quantum} bounds the optimal value of $C$ by optimizing the cost function landscape $\tilde C:\mathbb R^{2r}\rightarrow \mathbb R$, with
\begin{equation}\label{eq:cost_function_landscape}
  \tilde C(\bgamma,\bbeta)=\bra{\bgamma,\bbeta}H \ket{\bgamma,\bbeta},
\end{equation}
by a classical heuristic optimization method of choice. Here, $\ket{\bgamma,\bbeta}$ is the ansatz quantum state, depending on $2r$ parameters. It is prepared on a quantum computer by
\begin{equation}\label{eq:ansatz}
  \ket{\bgamma,\bbeta}=\prod_{k=r}^1S(\gamma_k,\beta_k)\ket +^{\otimes n},
\end{equation}
with $\ket + = (\ket 0 + \ket 1)/\sqrt{2}$, and the cycle
\begin{align}\label{eq:cycle}
S(\gamma_k,\beta_k)=&\exp\left( -\ii \frac{\beta_k}{2} \sum_{i=1}^n X_i\right)
 \exp\left( -\ii\frac{\gamma_k}{2} H\right)\nn\\
 =& RX^{\otimes n}(\beta_k)\left[\prod_iRZ_i(\om_i\gamma_k)\right]\nn \\ & \times \left[\prod_{i<j} RZZ_{ij}(\om_{ij}\gamma_k)\right].\
\end{align}
Here, $RX(\al)=\ee^{-\ii \al X/2}$, $RZ_i(\al)=\ee^{-\ii \al Z_i/2}$, and $RZZ_{ij}(\al)=\ee^{-\ii \al Z_iZ_j/2}$. In this work, the order of the products is from left to right, so that in \eq{ansatz}, $S(\gamma_1,\beta_1)$ is the first cycle to act on $\ket +^{\otimes n}$. To obtain an estimate for $\tilde C$, the ansatz state is prepared repeatedly, each time measuring $H$ at the end of the circuit. Since $H$ is diagonal in the computational basis, this measurement can be performed by measurements in the computational basis and classical postprocessing of the measurement outcomes.

\begin{figure}
\includegraphics[scale=1.1]{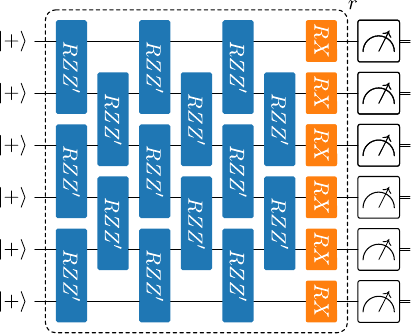}
\caption{SWAP-network implementation of the QAOA ansatz on $n=6$ qubits in the absence of 1-local terms in the problem Hamiltonian ($\om_i=0$), with $RZZ'=RZZ\cdot \SWAP$ and dependence on variational parameters suppressed. The ansatz cycle, marked by the dashed box, is applied to the initial state $\ket +^{\otimes n}$ a total of $r$ times, each time with new parameters. The cost function landscape is obtained through repeated preparation, measurement, and classical postprocessing.\label{fig:QAOA}}
\end{figure}

On NISQ devices, the ansatz state will be described by a mixed state $\rho(\bgamma,\bbeta)$. In this work, the focus is on correlated errors during the ansatz circuit and we will therefore assume perfect estimation of $\tilde C$ throughout. That is, we use
\begin{equation}\label{eq:cost_function_landscape_mixed}
\tilde C(\bgamma,\bbeta)=\tr[\rho(\bgamma,\bbeta) H]
\end{equation}
to compute the cost function landscape in our classical simulation of QAOA (\sec{numerical_methods_and_results}).

For SK problem Hamiltonians ($\om_{ij}\in\{1,-1\},\om_i=0$), each cycle of the ansatz circuit requires one two-qubit gate $RZZ$ between each pair of qubits. To accomplish this on quantum hardware without all-to-all connectivity, SWAP gates need to be inserted. On hardware with square-grid or line connectivity, the optimal way of doing so is by a SWAP network~\cite{kivlichan2018quantum, gorman2019generalized}. In the SWAP-network implementation of the ansatz circuit (\fig{QAOA}), the order of the product of $RZZ$ gates forms a brickwork structure, and a SWAP gate is inserted after each $RZZ$ gate \footnote{A SWAP gate effectively swaps qubit indices. This has to be accounted for by changing the parameters in subsequent $RZZ$ gates, which depend on the effective qubit labels the $RZZ$ gate is acting on, accordingly. One way of implementing this is by keeping a list of qubit indices, initialized at $x=(1,2,\ldots,n)$. Every time a SWAP gate acts on qubits $i$ and $j$, the list is updated with $x_i\leftrightarrow x_j$. Each $RZZ(\om_{ij}\gamma_k)$ has to be replaced with $RZZ(\om_{x_i x_j}\gamma_k)$ using the state of the list $x$ at the time the $RZZ$ gate acts. For an even number of cycles, the output state of the SWAP-network implementation is identical to the output of the circuit in \eq{ansatz}. For an uneven number of cycles, the effective order of the qubits is reversed, which can be easily accounted for in the classical postprocessing of the measurement outcomes.}. From $\om_{ij}\in \{1,-1\}$ and the antiperiodicity of $RZZ$ (and $RZZ'$) with period $2\pi$, it follows that the circuit cycle $S(\gamma_k,\beta_k)$ is invariant (up to an irrelevant overall sign) under $\gamma_k\mapsto \gamma_k\pm 2\pi$ for any $k$. As a consequence, the noiseless cost function landscape is fully periodic in $\gamma_k$ with period $2\pi$ for any $k$. A similar statement holds for $RX$. Because of the properties of the circuit ansatz and the SK Hamiltonian, much stronger and nontrivial symmetry relations hold, which are derived in \app{symmetries}. In the same appendix, we prove that these symmetries are not broken by spatiotemporally correlated Pauli error channels, which includes our fluctuator models if the unitary error operator $V$ (defined in Figs.~\ref{fig:error_model},~\ref{fig:error_model_spatial}) is a Pauli operator.

Returning to the general case [\eq{Hamiltonian}], a widely used figure of merit of optimization algorithms is the \emph{Approximation Ratio} (AR),
\begin{equation}\label{eq:AR}
 AR=\frac{\tilde C(\bgamma^*,\bbeta^*)}{C^*},
\end{equation}
with $\bgamma^*,\bbeta^*$ the optimal parameters as returned by the optimization algorithm, and $C^*$ the true, global optimum of the cost function $C$, which is equal to the ground state energy of $H$ by construction. The algorithm has found the global optimum of any $C$ if and only if $AR=1$. If, in contrast, the ansatz state is maximally mixed, $\rho(\bgamma,\bbeta)=\id/2^n$, which amounts to randomly guessing measurement outcome bit stings, $AR=0$ by the definition of $\tilde C$ [\eq{cost_function_landscape_mixed}] and the tracelessness of $H$ [\eq{Hamiltonian}]. In a noisy implementation of QAOA, the AR that is reached is in principle \emph{noise aware}; noise may cause the optimum of $\tilde C$ to change and the optimal parameters may adapt to this change by classical heuristic optimization.

To assess the resilience of QAOA against correlated errors, in this work we also consider the \emph{noise-unaware} AR. In classical simulation, where error probabilities can be changed at will, QAOA can be run without noise, obtaining the optimal parameters $\bgamma^{*}_0,\bbeta^{*}_0$. These parameters are thus unaware of the effects of noise on the cost function landscape. Putting them into the noisy circuit, we obtain $\tilde C(\bgamma^{*}_0,\bbeta^{*}_0 )$. Thus, we define the \emph{noise-unaware AR} as
\begin{equation}\label{eq:AR0}
 AR_0=\frac{~\tilde C(\bgamma^{*}_0,\bbeta^{*}_0 )}{C^*}.
\end{equation}
Using the noise-unaware AR, we define the \emph{noise adaptivity} as
\begin{equation}\label{eq:noise_adaptivity}
  \Delta AR = AR-AR_0.
\end{equation}
Given that the optimum found in the noiseless case and the optimum found in the noisy case are global, we have $\Delta AR\geq 0$, where $\Delta AR\neq 0$ indicates that the optimal parameters have adapted to the noise.

A quantity similar to $AR_0$ was introduced in Ref.~\cite{fontana2021evaluating}. There, instead of the expectation value of the Hamiltonian, the fidelity between the ansatz state and a target state was used as the cost function. Furthermore, the optimal noise-aware parameters were obtained by using the noise-unaware parameters as the initial point of optimization. This led the authors of Ref.~\cite{fontana2021evaluating} to call their noise-aware parameters the \emph{reoptimized} parameters. The advantage of reoptimization is that $\Delta AR\geq 0$ is guaranteed. In the current work, we find the optimal parameters for $AR$ and $AR_0$ independently to prevent any bias of the noise-aware optimal parameters to points close to the noise-unaware parameters. Furthermore, our data (\sec{numerical_methods_and_results}) obeys $\Delta AR\geq 0$ without enforcing this property by reoptimization.

\section{Error model}\label{sec:error_model}
The internal time evolution of a single fluctuator can be described by a two-state, discrete-time, and time-homogeneous Markov chain. (See, e.g., Refs.~\cite{breuer2007theory,serfozo2009basics} for background information on Markov chains.) Using this formulation, we obtain the correlation time, time-local marginalized error probabilities, and the expected number of errors of the error models used in this paper (Figs.~\ref{fig:error_model},\ref{fig:error_model_spatial}). Analytical first-order effects of \emph{uncorrelated} errors in the context of variational quantum algorithms were studied before in Refs.~\cite{xue2021effects,marshall2020characterizing,fontana2021evaluating}. We extend these methods to derive the analytical first-order effects of our correlated error models on QAOA.

In this section, we use terminology from the model of temporally correlated errors. Nonetheless, the results hold for the spatial fluctuator model if the state of the fluctuator after circuit layer $i$ is reinterpreted as the state of the fluctuator after qubit $i$. In this case, the correlation time $\tau$ becomes the correlation length $\la$.

\subsection{Two-level fluctuator}
Consider a single classical two-level fluctuator with ground state $(1,0)^\mrm{T}$ and excited state $(0,1)^\mrm{T}$, with $\mrm T$ the transpose. If the fluctuator is excited with probability $0\leq p\leq 1$, its state is described by the ensemble
\begin{equation}
s_0=\left( \begin{array}{c} 1-p \\ p\end{array} \right).
\end{equation}
This does not represent a quantum state, but a classical ensemble equivalent to the classical mixed quantum state $\mrm{diag}(s_0)$. We take $s_0$ as the initial state of the fluctuator. Define the random variable
\begin{align}\label{eq:RV}
  B[(1,0)^\mrm{T}]=0
  ,&&B[(0,1)^\mrm{T}]=1,
\end{align}
which can, e.g., be imagined as the strength of a magnetic field caused by the fluctuator at its associated qubit.

At time $t$ (in units of gate time), the state of the fluctuator, $s_{t}$, is retained with probability $0\leq \kappa \leq 1$, and it is reset to $s_0$ with probability $1-\kappa$. We refer to $\ka$ as the correlation strength. After this process, the state is given by $s_{t+1}=Ts_{t}$, with $T$ the transition matrix,
\begin{equation}\label{eq:transition}
  T=\left[
    \begin{array}{cc}
      \kappa+(1-\kappa)(1-p)  & (1-\kappa)(1-p)\\
      (1-\kappa)p          & \kappa+(1-\kappa)p
    \end{array}
    \right],
\end{equation}
denoted by $T_\kappa$ in Figs.~\ref{fig:error_model}, \ref{fig:error_model_spatial}.

To see that $T$ describes the desired process, assume, for example, that at time $t$, the fluctuator is excited; $s_t=(0,1)^\mrm{T}$. After it is reset with probability $1-\kappa$, there are two ways in which it can remain excited: either the fluctuator was not reset, which happens with probability $\kappa$, or the fluctuator was reset to $s_0$, but is excited merely because in the state $s_0$, the fluctuator is excited with probability $p$. Thus, the overall probability that the fluctuator remains excited is  $(0,1)T(0,1)^\mrm{T}=T_{11}=\kappa+(1-\kappa)p$. Other entries of $T$ follow similarly. Note $T_{ba}$ describes the probability to transition from the state $a$ to $b$ (in terms of the values that the random variable $B$ can take).

The initial state $s_0$ is a steady state of $T$ because $T$ sends $s_0$ to $s_0$ both when the state is kept and when it is reset. (Explicit diagonalization is a simple way to check that this is the only nontrivial steady state). Thus, from an ensemble viewpoint, the Markov chain is trivial; $s_t=T^ts_0=s_0$ for all $t$. Therefore, the marginalized, time-local probability to be in the excited state is $p$, independent of $t$ and $\kappa$. Under the definition $B_t=B(s_t)$, it follows that $\mathbb E(B_t)=p$ for all $t$. Defining the random variable $B_{\mrm{tot}}=\sum_{t'=0}^tB_{t'}$, it follows by the linearity of expectation values, which also holds for correlated random variables, that the expected number of times the fluctuator has been excited from $t=0$ up to and including $t=t'$ is simply $(t+1)p$, irrespective of $\kappa$.

The state of the fluctuator at time $t$ is given by $s_t=T^ts_0$, where the $t$th matrix power of $T$ is  given by
\begin{equation}
  T^t=\left(
    \begin{array}{cc}
      1-p+p\kappa^t&1-p-\kappa^t+p\kappa^t\\
      p-p\kappa^t &p+\kappa^t-p\kappa^t
    \end{array}
   \right).
\end{equation}
Note that \eq{transition} is retrieved at $t=1$. With a concise expression for $T^t$ at hand, the correlator $\mc C(\Delta t)=\mathbb E(B_tB_{t+\Delta t})-\mathbb E(B_t)\mathbb E(B_{t+\Delta t})$, which does not depend on $t$ by time homogeneity of the Markov chain, is computed straightforwardly, yielding
\begin{equation}\label{eq:correlation_function}
  \mc C(\Delta t)=4p(1-p)\,\kappa^{\lvert \Delta t \rvert}.
\end{equation}
That is, correlations decay exponentially, with $1/\ee$ correlation time
\begin{equation}\label{eq:correlation_strength}
  \tau=-\frac{1}{\ln \kappa}.
\end{equation}
Thus, the correlation time increases monotonically from $0$ to $\infty$ as the correlation strength $\kappa$ is increased from $0$ to $1$. (In the spatial fluctuator model, the correlation length $\la$ increases monotonically from $0$ to $\infty$ as $\kappa$ is increased from $0$ to $1$.)   

Even though $s_t=T^ts_0=s_0$ for all $t$, in the course of each physical run of a quantum circuit involving $m$ circuit layers, a nontrivial \emph{realization} $R\in\{(1,0)^\mrm{T},(0,1)^\mrm{T}\}^{m+1}$ of fluctuator states is sampled. Equivalently, we can say that, during each run, a realization bit string $b\in\{0,1\}^{m+1}$ of outcomes is obtained, where the $t$th entry of $b$, $b_t\in\{0,1\}$, denotes the outcome of a trial of the random variable $B_t$. The realization $b$ occurs with probability
\begin{equation}\label{eq:pb} p_b=[\delta_{b_0,0}(1-p)+\delta_{b_0,1}p]\prod_{t=0}^{m-1}T_{b_{t+1}b_t},
\end{equation}
using $p_b$ as shorthand notation for $p(B_0=b_0,B_1=b_1,\ldots)$. With the above expression, it can be seen with an explicit calculation that, also from the realization viewpoint, the marginalized, time-local probability that the fluctuator is excited at time $t$ is
\begin{equation}
  p(B_{t}=1)=\sum_{\{b\,:\,b_{t}=1\}}p_{b}=p,
\end{equation}
independent of $t$ and $\kappa$.

\subsection{Circuit fidelity}
Consider the circuit fidelity $F=\bra{\psi} \rho \ket{\psi}$ between the noiseless output state $\ketbra{\psi}{\psi}$ of a VQA's quantum circuit and the noisy output state $\rho$ of that same quantum circuit. If we assume that a single error operator $V$ (Figs.~\ref{fig:error_model},~\ref{fig:error_model_spatial}) occurring during the circuit leads to an output state that is orthogonal to the noiseless output state of that circuit, we have in the temporal fluctuator model that $F=(p_{0^{m}})^n$. Here, $p_{0^{m}}$ is the probability that no error occurs during the $m$ layers of the circuit, and $n$ is the number of qubits. The probability $p_{0^{m}}$ is computed straightforwardly with \eq{pb} \footnote{Because of the way we have defined the interaction between the fluctuator and the qubit, there are two ways in which no errors can occur; $p_{0^m}=\sum_{b_0}p_{b_0 0^m}$. By \eq{pb}, we thus have $p_{0^{m}}=(1-p)[1-p(1-\kappa)]^{m}+p(1-\kappa)(1-p)[1-p(1-\kappa)]^{(m-1)}$.}. For uncorrelated errors, we have $F=(1-p)^{mn}$, whereas for fully correlated errors, we have $F=(1-p)^n$. Thus, one may expect an increasing circuit fidelity with increasing correlation strength, and therefore an \emph{improvement} of the performance of VQAs with increasing correlation strength.

The corresponding formulas for the circuit fidelity $F$ in the spatial fluctuator model are obtained by interchanging $m$ and $n$. Thus, in the case of strong correlations and deep circuits ($m>n$), we expect a lower circuit fidelity in the spatial fluctuator model than in the temporal fluctuator model. Therefore, in the same case, we may expect a worse performance of VQAs in the spatial fluctuator model than in the temporal fluctuator model. Because $mn>m,n$, we expect an even worse performance in the case of fully uncorrelated errors.

However, since one error operator $V$ may cause an output state that is not exactly orthogonal to the noiseless output state, the above discussion merely puts a lower bound on the circuit fidelity. Furthermore, even if the circuit fidelity is exactly known, it merely puts a bound on the AR, which may be loose~\cite{beach2019making}. Additionally, it is not clear a priori how the variational parameters may change due to the noise and the correlations thereof. Finally, because of the breakdown of quantum error correction in the presence of strongly correlated noise, one might expect a \emph{decrease} in the performance of VQAs with increasing correlation strength. In the next section, we therefore derive analytical expressions for the effects of the fluctuator error models on the AR. These are valid in the regime of small error probability and do not account for any possible noise-induced change in the optimal variational parameters. In the section thereafter, we show numerical results that are valid outside the regime of small error probability and account for a possible noise-induced change in optimal variational parameters.

\subsection{Expectation values}\label{sec:expectation_values}
Directly after the fluctuator has transitioned from $b_{t-1}$ to $b_t$, the unitary error operator $V$ is applied to the associated qubit, conditioned on the state of the fluctuator. The transition and the controlled unitary do not change the reduced state of the fluctuator, but correlations are built up by the controlled unitary nevertheless. At the end of the circuit, all fluctuators are traced out. Let $U_1,\dots,U_m$ be the layers of a quantum circuit (where $U_1$ denotes the layer that is applied first) and let $\rho$ be its initial state. In QAOA, these layers will depend on the variational parameters $\bgamma,\bbeta$, but in this section, their dependence on these parameters is mostly suppressed.  The effect of a single fluctuator $f$ in the temporal fluctuator model, at the end of the circuit and  after tracing out the fluctuator, is described by the noisy circuit
\begin{align}\label{eq:noisy_circ}
  \tilde{\mc U}_f(\rho_0)=\sum_{b^f} p_{b^f}\, \tilde U_{b^f}^\nodagger \rho_0\, \tilde U^\dagger_{b^f}, &&
  \tilde U_{b^f}=\prod_{t=m}^{1}V_{b_t^f} U_t,
\end{align}
with $p_{b^f}$ as in \eq{pb} (after substituting $b\rightarrow b^f$), $\tilde U_{b^f}$ the circuit in case of noise realization $b^f$, and $\rho_0=(\ketbra{+}{+})^{\otimes {n}}$ the initial state. Here, $V_{b_t^f}$ is an operator acting immediately after time $t$ and immediately before the gates $U_{t+1}$ on the qubit associated with $f$, with $V_0\equiv\id$, $V_1\equiv V$. [In \eq{noisy_circ}, the product ordering is such that $U_1$ appears on the right. When $\tilde U_{b^f}$ is applied to a state, the first operator that acts on that state is $U_1$.]

Because the $n$ realization probability distributions $\{p_{b^f}\}_f$ are independent and identically distributed, their combined effect on the output state is
\begin{equation}\label{eq:noisy_circ_all}
  \tilde{\mc U}(\rho_0)=\sum_b p_b \tilde U_b^\nodagger \rho_0\, \tilde U_b^\dagger,
\end{equation}
with
\begin{equation}\label{eq:puv}
p_b=\prod_{f=n}^{1}(p_{b^f}),\ \tilde U_b=\prod_{t=m}^{1}(V_{b_t} U_t),\ V_{b_t}=\prod_{f=n}^1(V_{b_t^f}).
\end{equation}

In the spatial fluctuator model, we again have \eq{noisy_circ_all}, but with
\begin{equation}\label{eq:noisy_circ_spatial}
p_b=\prod_{t=m}^{1}(p_{b^t}),\ \tilde U_b=\prod_{t=m}^{1}(V_{b^t} U_t),\ V_{b^t}=\prod_{q=n}^1(V_{b_q^t}).
\end{equation}
In the spatial model, the different fluctuators carry the label $t$. Fluctuator $t$ transitions from the state $b^t_q$ to $b^t_{q+1}$ as it `moves' from qubit $q$ to $q+1$.

Returning to the temporal fluctuator model, we define the  susceptibility of the cost function to noise at optimal variational parameters (as found by QAOA) as  $\chi=(\dd \langle H \rangle /\dd p)_{p=0}$, with $\tr[\rho(\bgamma^*_0,\bbeta^*_0)H]\equiv \langle H \rangle$ for short. Using $\chi$, we may write the first order approximation to the AR, which we call the \emph{linearized AR}, as
\begin{equation}\label{eq:lin}
  AR^\mrm{lin}(p)=AR(0)+p \frac{\chi}{C^*},
\end{equation}
where we have by \eq{noisy_circ} that
\begin{equation}\label{eq:dhdp}
  \chi=\sum_b\left. \frac{\dd p_b}{\dd p}\right\rvert_{p=0}\tr\left(\tilde U_b^{\nodagger}\rho_0\,\tilde U^\dagger_b H \right).
\end{equation}
In this section, we are mainly interested in $\chi$, for which no approximations are made. The linearized AR, which approximates the AR up to an error $O(p^2)$ as $p$ goes to zero, is introduced for later reference.

Since $p_b$ [\eq{pb}] is some explicit polynomial in $p$ and $\kappa$, it is clear that the derivative with respect to $p$ at $p=0$ can be computed analytically. The resulting expression for $\chi$ is derived in \app{susceptibility}. The expression becomes especially clear in the limits of no correlation ($\kappa=0$) and full correlation ($\kappa=1$),
\begin{equation}\label{eq:chi}
\chi_\kappa^{(a)}=\lvert \mc B_{1+a\delta_{\kappa 1}}\rvert \{\langle H \rangle^{(1+a\delta_{\kappa 1})}-\langle H \rangle_{\boldsymbol 0}\},
\end{equation}
where $a=m$ for fully temporally correlated errors, and $a=n$ for fully spatially correlated errors. Here, $\mc B_\ell$ is the set of realizations $b$ where exactly one fluctuator has exactly one chain of $\ell$ contiguous excitations (no other excitations in any fluctuator present), $\langle H \rangle^{(\ell)}$ is the expectation value of $H$ given realization $b$, averaged over all $b\in\mc B_\ell$, and $\langle H \rangle_{\boldsymbol 0}$ is the expectation value of $H$ in the noiseless case.

For fully temporally correlated errors, we have the prefactor $\lvert \mc B_{1+m}\rvert=n$, for fully spatially correlated errors, $\lvert \mc B_{1+n}\rvert=m$, and for fully uncorrelated errors $\lvert \mc B_{1}\rvert=n(m+1)$. So, roughly speaking, there are two effects that determine whether we should expect $\chi_0>\chi_1$ or $\chi_0<\chi_1$. On the one hand, one expects that ($i$), on average, the detrimental effect of a fully correlated error, consisting of $m$ consecutive errors ($n$ adjacent errors in the case of spatially correlated errors) on $\langle H \rangle$ is larger than the detrimental effect of 1 error on $\langle H \rangle$. That is, one expects $(\langle H \rangle^{(m+1)} - \langle H \rangle_{\boldsymbol 0})>(\langle H\rangle^{(1)}-\langle H \rangle_{\boldsymbol 0})$ and $(\langle H \rangle^{(n+1)} - \langle H \rangle_{\boldsymbol 0})>(\langle H\rangle^{(1)}-\langle H \rangle_{\boldsymbol 0})$. On the other hand, ($ii$) there are fewer ways in which a fully correlated error can happen: a fully temporally correlated error can happen in $n$ different ways and a fully spatially correlated error can happen in $m$ different ways, whereas a single uncorrelated error can happen in $nm$ ways. [The latter still leads to a prefactor of $\lvert \mc B_1\rvert=n(m+1)$ because of the way $\langle H \rangle^{(1)}$ is defined.] If the effect ($ii$) is dominant, we expect that for deep circuits ($m>n$), the AR is highest with fully temporally correlated errors, followed by the AR with fully spatially correlated errors, followed by the fully uncorrelated errors. This order of ARs would be consistent with the expectations that follow from the fidelity arguments in the previous section.

\section{Numerical methods and results}\label{sec:numerical_methods_and_results}
In this section, we present numerical data on the SWAP-network implementation of QAOA (\fig{QAOA}) under the influence of the temporal and spatial fluctuator models. As the error operator $V$ (Figs.~\ref{fig:error_model},~\ref{fig:error_model_spatial}), we choose the bit-phaseflip error or Pauli-$Y$ error. This choice is motivated by the fact that $Y$ is the only nontrivial single-qubit Pauli operator that commutes with neither $RZZ'$ nor $RX$. To exclude the possibility that an increase in AR as a function of $\kappa$ is due to $Y^2=\id$, which may occur, for example, due to $Y$ operators being inserted at $t=1$ and $t=2$ on the first qubit, we include an interaction between the fluctuator and its associated qubit only after the qubit was acted on by a $RZZ'$ or $RX$ gate. [Since in SK models $\om_i=0$, the entire circuit (excluding measurement) consist of $RZZ'$ and $RX$ gates, as in \fig{QAOA}.]

Simulations were performed for 16 random SK instances on $n=6$ qubits, with $r=3$ cycles (leading to $m=21$ circuit layers) and various values of $p$ and $\kappa$. For each data point, that is, for each combination of SK instance, type of fluctuator model, $p$ and $\kappa$, a basin-hopping routine~\cite{wales1997global} was run 32 times to heuristically optimize the cost function landscape $\tilde C$. The meta-parameters of the routine were the default values as per SciPy~\cite{scipybasinhopping}, but with 4 iterations and all initial parameters $\bgamma,\bbeta$ chosen at random in the interval $[-0.5\cdot 10^{-3},0.5\cdot 10^{-3})$.  The randomness of the optimization routine on the one hand, and the regularity of the outcomes on the other, indicates that the optimization routine consistently found the global optimum, as is discussed in detail in the Supplemental Material \footnote{All code and data used to generate the results in this work are presented in the Supplemental Material [URL will be inserted by publisher] and on \url{https://github.com/kattemolle/HQAOA}.}.

The mixed state $\rho'$ at the end of the quantum circuit was obtained by full density-matrix simulation of all qubits and fluctuators, tracing out all fluctuators before computing the cost function $\tilde C=\tr(\rho' H)$ (dependence on variational parameters suppressed). With the full density matrix at hand, the computation of the cost function at given $p,\kappa$ and variational parameters is free of shot noise and thus essentially exact. The simulator itself was obtained by extending the simulator used in Ref.~\cite{kattemolle2021variational} with the functionality of density-matrix simulation of mixed quantum-classical registers. All code used to generate the results in this work, including the code of the simulator, is available as Supplemental Material~\cite{Note3}.

We find qualitatively equal results for all 16 random problem instances. In this section, we report on the results for one typical instance. Details on its typicality and on the other 15 instances are found in Appendix~\ref{sec:instance_properties} and the Supplemental Material~\cite{Note3}.

We first present the numerical results regarding the noise-adaptivity, followed by our results on the dependence of the AR on the correlation strength $\ka$. Figure~\ref{fig:p_plot} shows the AR as a function of $p$ at fixed $\kappa\in\{0,1\}$ for the typical instance, obtained under the temporal and spatial fluctuator models. Increasing $p$ away from $0$, we observe a linear decrease of the AR, with excellent agreement between the AR, the noise-unaware AR, and the linearized AR. In this sense, no robustness of QAOA against errors of the three types is observed around $p=0$.

\begin{figure}
  \begin{center}
    \includegraphics{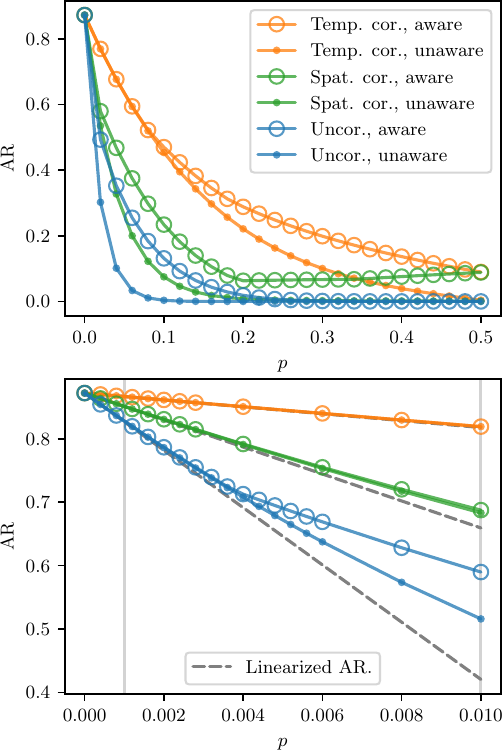}
  \end{center}
  \caption{\label{fig:p_plot} The approximation ratio (AR), obtained by the noisy SWAP-network implementation of QAOA for a typical problem instance (see Appendix~\ref{sec:instance_properties}), for $0\leq p \leq 0.5$ (top) and $0\leq p\leq0.01$ (bottom), at fixed correlation strength $\ka$. Results are shown for fully temporally correlated errors (orange), fully spatially correlated errors (green), and fully uncorrelated errors (blue). Open circles show the noise-aware AR [\eq{AR}], whereas the filled circles show the noise-unaware AR [\eq{AR0}]. The linearized AR [\eq{lin}] is displayed as a dashed gray line for visibility, irrespective of the noise model and correlation strength. Light gray lines at $p\in\{0.001,0.01\}$ are for later reference. Error bars are absent because of the use of full density-matrix simulation. The AR in the case of random guessing is exactly zero (\sec{QAOA}).}
\end{figure}

For uncorrelated noise, a divergence between the AR and the noise-unaware AR is observed at $p_\mrm{crit}^\mrm{uncor}=(3.9\pm 0.1) \times 10^{-3}$. This coincides with an abrupt jump in the otherwise slowly varying optimal parameters (see Supplemental Material~\cite{Note3}). This shows that, as $p$ is increased away from $0$, initially no other local minimum of the cost function becomes available and that the location of the initial minimum is approximately constant. For $p>p_{\mrm{crit}}^{\mrm{uncor}}$, for the first time, a remote lower local minimum becomes available. This new local minimum remains the lowest local minimum as found by QAOA until $p=0.015\pm 0.005$. Similar first critical points, not visible in the plot because of the scale, occur at  $p_{\mrm{crit}}^{\mrm{temp}}=0.077\pm0.001$ for fully temporally correlated noise, and at $p_{\mrm{crit}}^{\mrm{spat}}=0.013\pm0.001$ for fully spatially correlated noise. For noise that would break the symmetry of the cost function landscape, a divergence between the noise-aware AR and the noise-unaware AR would already be expected at $p=0$ \cite{fontana2021evaluating}.

An increase of the noise-aware AR in the case of fully spatially correlated errors is seen after $p\approx 0.2$. This effect may arise because the noise creates a mixture of states, some of which may have a favorable cost. Then, increasing the error probability may in some cases lead to an increase in the AR. This effect is perhaps best illustrated by a simple example. Consider a trivial version of QAOA, with a single qubit in the initial state $\ket +$, no variational circuit, and Hamiltonian $H=Z$, under the effect of the error channel $\rho\mapsto (1-p)\rho+RY(\pi/2)\, \rho\,RY^\dagger(\pi/2)$. Then, with $RY(\pi/2)\ket +=\ket 1$, the cost function `landscape' after the error channel is $\tilde C=-p$. That is, in this example, $AR(p)=p$.

We now discuss the numerical data and results on the dependence of the AR on the correlation strength. Our conclusion that error correlations can be beneficial for VQAs such as QAOA is based on these data, together with the analytical results of \sec{expectation_values} and the additional data in the Supplemental Material~\cite{Note3}. In addition to data on uncorrelated errors, \fig{p_plot} shows data on the extreme cases of fully temporally and fully spatially correlated errors. Out of the three types of correlation (uncorrelated, temporal, spatial), the AR is highest for fully temporally correlated errors. Second in AR comes the fully spatially correlated errors, followed by the AR for uncorrelated errors. This shows that, for the cases studied here, effect ($ii$) in \sec{expectation_values} is dominant for small $p$. The order of the ARs remains unchanged for all $p$ considered, indicating that the same effect plays a role for all $p$ \footnote{A single exception occurs for one instance at $p\geq 0.4$. Here, the noise-aware AR in the case of fully spatially correlated errors becomes higher than the noise-aware AR in the case of fully temporally correlated errors. The noise-aware AR remains higher for correlated errors. See the Supplemental Material for details.}.

After having discussed these extreme cases, we now take a closer look at the effects of partially correlated errors. Figure~\ref{fig:kappa_plot} shows the AR as a function of $\kappa$ at fixed $p\in\{0.001,0.01\}$ for both fluctuator models. These values for $p$ are chosen because typical two-qubit error probabilities roughly fall within the range $0.001 \lesssim p \lesssim 0.01$~\cite{ballance2016high,xue2022quantum,kjaergaard2020superconducting}. The data suggest that the noise-aware AR is highest for temporally correlated errors, and that the increase in the AR as a function of $\kappa$ is monotonic for all fluctuator models and $p$ considered. At $p=0.001$, excellent agreement between the noise-aware AR, the noise-unaware AR, and the linearized AR is shown for all $0\leq\kappa\leq 1$ considered. At $p=0.01,\kappa=0$, the agreement between the AR, the noise-unaware AR, and the linearized AR has been broken, as was already apparent from \fig{p_plot}. In \fig{kappa_plot}, we see how the first critical $p$ depends on $\kappa$; $p=0.001$ is below $p_{\mrm{crit}}^\kappa$ for both fluctuator models and any $0\leq\kappa\leq 1$. In the spatial fluctuator model, $p=0.01$ is \emph{above} $p_{\mrm{crit}}^{\mrm{spat},\kappa}$ for all $0\leq\kappa\leq 1$. In the temporal fluctuator model, $p=0.01$ is below $p_{\mrm{crit}}^{\mrm{temp},\kappa}$ for $\kappa$ below roughly~$0.7$, and above  $p_{\mrm{crit}}^{\mrm{temp},\kappa}$ for $\kappa$ above roughly~$0.7$.

\begin{figure}
\vspace{1em}
  \begin{center}
\noindent\begin{minipage}{0.73\columnwidth}
  \includegraphics{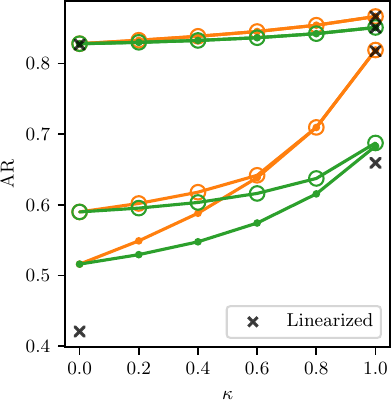}
\end{minipage}
\hfill
\begin{minipage}{0.25\columnwidth}\raggedright
  \vspace{.6em}
$\left.\begin{array}{l}
         \vspace{.5em}
\end{array}\right\}p=0.001
$
\vspace{.5em}
\noindent$\left.\begin{array}{l}
                  \vspace{14em}
\end{array}\right\}p=0.01
$ \\
\vspace{3.5em}
\end{minipage}
\end{center}
\vspace{-1em}
      \caption{\label{fig:kappa_plot} The approximation ratio (AR) for a typical problem instance (see Appendix~\ref{sec:instance_properties}), as a function of the correlation strength $\kappa$, at fixed $p=0.001$ (upper lines and crosses) and $p=0.01$ (lower lines and crosses). Note that the correlation time $\tau$ and the correlation length $\la$ are monotonically increasing functions of $\kappa$ [\eq{correlation_strength}]. At $\kappa=0$, errors are fully uncorrelated, whereas at $\kappa=1$, errors are fully correlated. The values for $p$ correspond to the vertical lines in \fig{p_plot}. Definitions are as in \fig{p_plot}, except for the linearized AR, which is now displayed with crosses. At $\kappa=0$, the temporal fluctuator model is equivalent to the spatial fluctuator model, and hence only two black crosses are shown at $\kappa=0$, the upper one showing the linearized AR at $p=0.001$, and the lower one showing the same at $p=0.01$. For visibility, at $\kappa=1$, crosses are black irrespective of the type of fluctuator model and $p$. From top to bottom, they describe the linearized AR at $p=0.001$ for the temporal fluctuator model and spatial fluctuator model, followed by the linearized AR for the temporal fluctuator model and spatial fluctuator model at $p=0.01$.}
\end{figure}

\section{Discussion}\label{sec:discussion}
We studied the performance of QAOA under temporally and spatially correlated errors using a physically inspired toy error model. In the model for temporally correlated errors, each qubit interacts with one independent classical fluctuator during a quantum computation. In the model for spatially correlated errors, after each gate time, all qubits interact with a single fluctuator that is reset before the next gate time. As the fluctuator moves through time (temporal model) or space (spatial model), it undergoes an internal time evolution described by a Markov process. Using the Markov formulation, we showed that the spacetime-local marginalized error probability $p$ is independent of the correlation strength of the Markov chain, and how the latter can be varied from zero to infinity.

We showed that, to first order in the local error probability $p$, the effect of the fluctuator error models on QAOA's cost function landscape has two competing factors. On the one hand, ($i$) the detrimental effect of a single correlated error may be worse than the effect of a single uncorrelated error, but on the other, ($ii$) there are far fewer ways in which a correlated error may act. For example, a fully temporally correlated error (that is, the error operation $V$, defined in Figs.~\ref{fig:error_model} and \ref{fig:error_model_spatial}, is inserted after each gate time) during a circuit acting on $n$ qubits can happen in $n$ ways, a fully spatially correlated error (i.e., an error operation $V$ is inserted at every qubit) during a circuit of depth $m$ can happen in only $m$ ways, whereas a fully uncorrelated error (i.e., a single error operation $V$ is inserted anywhere in the circuit) can happen in $nm$ ways.

We numerically simulated QAOA for 16 random SK problem instances on 6 qubits. For all instances, we observed an increase in the performance of QAOA with correlation strength at all $p$, indicating that effect ($ii$) is dominant \footnote{There are two instances where a single exception occurs, but these exceptions can be attributed to incomplete optimization. See the Supplemental Material for details.}. As a separate effect, in all instances and at all correlation strengths, we observed the existence of critical local error probabilities after which the noise-aware approximation ratio is higher than the noise-unaware approximation ratio. This shows QAOA can in some cases adapt to the effects of uncorrelated and correlated errors, given that error probabilities are \emph{above} a threshold. It remains an open question how this threshold behaves as the number of qubits is increased to the numbers required for a practical quantum advantage. The noise resilience of QAOA may be seriously limited if the noise-adaptivity threshold increases with $n$.

Our work does not show that correlated errors are beneficial to QAOA; adding correlated errors \emph{on top} of local errors will almost certainly result in a reduction in the performance of QAOA. What our work indicates is that it is not correlation in itself that has a negative effect on QAOA. To improve current hardware to the point that it can demonstrate a useful quantum advantage, a reduction of both uncorrelated and correlated errors remains necessary.

The error model used in this work is minimalistic, but we do not expect that the qualitative results of our work depend on the details of more complicated or realistic noise models. A first limitation of our noise model is that it is ultimately a stochastic unitary error model, where unitary error operations are inserted into the circuit according to some (correlated) probability distribution [\eq{noisy_circ_all}]. Not all quantum channels are of this type, such as the amplitude damping channel, but they can be transformed into an effective stochastic error model (specifically, a stochastic Pauli error model) by the process of Pauli twirling \cite{berg2022probabilistic,flammia2022averaged}. As a result, if QAOA is run on a real quantum device with added Pauli twirling operations, the physical noise experienced during the quantum computation will effectively be formed by stochastic Pauli errors, thereby also reinstating the symmetries of the SK QAOA cost function landscape derived in \app{symmetries}.

A second limitation of our noise model is that all stochastic error unitaries inserted into the circuit are equal. Although allowing, e.g., gate-dependent errors would result in a more realistic noise model, we expect that our results would not be qualitatively affected. This is because we expect the states during the ansatz circuit to be sufficiently spread out over the Hilbert space, and to vary enough from circuit layer to circuit layer, that the choice of error operation $V$ is irrelevant as long as it does not fully commute with the gates in the ansatz circuit. If, on the other hand, the error operator were to commute with, for example, the $RZZ$ gates, it would effectively lead to one layer of error operations per circuit cycle, coarsening the error model.

A final limitation of our error model is that it only considers error correlations with an exponentially decaying correlation function [\eq{correlation_function}]. This does not include scenarios in which the correlations decay as a power law. However, it is important to note that if our model were to be expanded to accommodate power-law correlations, it would be equivalent to the current model in the cases of uncorrelated and fully correlated errors, therefore leaving \fig{p_plot} unaltered. The same holds for errors with a $1/f$ spectral density. Additionally, $1/f$ noise has a correlation function proportional to $\mrm{Ei}(-f_0 \lvert t\rvert)=O(\ee^{-f_0\lvert t \rvert})$ \cite{burkard2009non-markovian}, with $f_0$ a low-frequency cutoff. Hence, a low-frequency cutoff introduces a timescale that confines correlations to within a finite time span, much like the correlations in our model.

Lines for future work include the study of the generality of our results. It is a priori not clear if similar effects hold in more general models of temporally and spatially correlated noise, or if they hold for other VQAs. If they do, the experimental requirements on noise correlation strengths may not be as stringent as previously thought, bringing a useful quantum advantage within closer reach.\\

Code and data are available as Supplemental Material~\cite{Note3}.
\begin{acknowledgements}
We acknowledge funding from the Competence Center Quantum Computing Baden-W\"urttemberg, under the project QORA. Classical simulations were carried out on the Scientific Compute Cluster of the University of Konstanz (SCCKN).
\end{acknowledgements}

\bibliography{bib.bib}

\appendix

\section{Cost function landscape noise susceptibility}\label{sec:susceptibility}
Here, we analytically compute $\chi$ [\eq{dhdp}], the derivative of the cost function landscape $\tilde C$ [\eq{cost_function_landscape}] (at fixed variational parameters) to the local error probability $p$ at $p=0$.  The only assumption on $\tilde C$ is that it is obtained using a circuit that is acted upon by our fluctuator models with noise parameters $p$ and $\kappa$ (Figs.~\ref{fig:error_model},~\ref{fig:error_model_spatial}).

Let us first consider a single fluctuator $f$. Denote the noise realizations of that fluctuator by the bit string $b^f$. Computing the derivative of \eq{pb}, we find
\begin{align}\label{eq:dpbdp}
  \left.\frac{\dd p_{b^f}}{\dd p}\right\rvert_{p=0}=&\sum_{\ell=1}^{m+1}\sum_{i=0}^{m+1-\ell}\delta_{b^f,0^i 1^\ell 0^{k_{i\ell}}}(1-\kappa)^{2-\delta_{i,0}-\delta_{k_{i\ell},0}}\kappa^{\ell-1}\nn\\
&-\delta_{b^f,\boldsymbol 0}[m(1-\kappa)+1],
\end{align}
with $k_{i\ell}=m+1-i-\ell$ the number of $0$s that need to be appended to the bit string $0^i1^\ell$ to make it a bit string of length  $m+1$.

Moving to $n$ fluctuators, we have by \eq{puv} that
\begin{equation}
  \left.\frac{\dd p_b}{\dd p}\right\rvert_{p=0}=\sum_{f=1}^n \delta_{b^1\ldots b^{f-1},\boldsymbol 0} \left(\frac{\dd p_{b^f}}{\dd p}\right)_{p=0}\delta_{b^{f+1}\ldots b^n,\boldsymbol0}.
\end{equation}
Then, with the definition $\langle H \rangle_b=\tr(\tilde U^\nodagger_b \rho_0\, \tilde U^\dagger_b H ) $, we have by \eq{dhdp}, that
\begin{align}
\chi=&\sum_{b,f,\ell,i}\tilde\delta_{b,f,\ell,i}(1-\kappa)^{2-\delta_{i,0}-\delta_{k_{i,\ell},0}}\kappa^{\ell-1}\langle H \rangle_b\nn\\
  &-n[m(1-\kappa)+1]\langle H \rangle_{\boldsymbol{0}},
\end{align}
with $\tilde \delta_{b,f,\ell,i}=\delta_{b^1\ldots b^{f-1},\boldsymbol 0} \delta_{b^f,0  ^i1^\ell 0^{k_{i\ell}}}\delta_{b^{f+1}\ldots b^n,\boldsymbol0}$, and where the sum is over all $b,i,j,f$, in the ranges as described before. Note $\boldsymbol 0=00\ldots 0$ has varying dimension depending on context. We can write the previous expression for $\chi$ more schematically as
\begin{align}
 \chi=\sum_b&\left[(1-\kappa)^{h(b)}\kappa^{\ell (b)-1}\langle H \rangle_b\right]\nn\\&-n[m(1-\kappa)+1]\langle H \rangle_{\boldsymbol 0},
\end{align}
where the sum is over all realizations $b$ where exactly one fluctuator has exactly one chain of consecutive excitations (and no other excitations), $h(b)\in\{0,1,2\}$ is the number of transitions $0\leftrightarrow 1$ in $b$, and $\ell(b)\geq 1$ is the length of the chain.

The average contribution of chains of errors of length $\ell$ to $\chi$ schematically reads $\langle H \rangle^{(\ell)}\equiv \frac{1}{\lvert \mc B_\ell \rvert}\sum_{b\in\mc B_\ell}(1-\kappa)^{h(b)}\langle H \rangle_b$, where the set $\mc B_\ell$ contains those realizations $b$ where exactly one fluctuator has exactly one chain of excitations of length $\ell$, and where we have absorbed the boundary effects $(1-\kappa)^{h(b)}$ into $\langle H \rangle ^{(\ell)}$. With this definition, the susceptibility becomes
\begin{equation}
\chi=\sum_{\ell=1}^{m+1}\left(\lvert \mc B_\ell \rvert \kappa^{\ell-1}\langle H \rangle^{(\ell)}\right)-n[m(1-\kappa)+1]\langle H \rangle_{\boldsymbol 0}.
\end{equation}

\section{Symmetries of the QAOA SK cost function landscape}\label{sec:symmetries}
Here, we prove that a certain set of transformations of the parameters of the QAOA SK cost function landscape generates a symmetry group of the cost function landscape. Note that the SK Hamiltonian is given by \eq{Hamiltonian} with $\om_{ij}\in\{1,-1\},\ \om_i=0$ and that hence the Hamiltonian itself possesses a global  $\mathbb Z_2$ symmetry. Additionally, because $\om_i=0$ the ansatz circuit [\eq{ansatz}, \fig{QAOA}] does not contain $RZ$ gates. The numerical data available in the Supplemental Material~\cite{Note3} suggest that the transformations derived in the current appendix generate the entire symmetry group of the cost function landscape. The symmetries are not broken under a general class of noise models that we call the spatiotemporally correlated Pauli noise channels. The main application of the results in this section is that they allow for the comparison of optimal variational parameters modulo the symmetries of the cost function landscape.

\subsection{Symmetry group}

The class of Pauli channels contains all channels of the form
\begin{equation}
\Lambda(\rho)=\sum_a p_a P^\nodagger_a \rho\, P^\dagger_a,
\end{equation}
with $p_a$ a probability distribution, and $\{P_a\}_a$ tensor products of the Pauli operators $\{\id,X,Y,Z\}$, also referred to as Pauli words. Note that Pauli channels have a stochastic interpretation: with probability $p_a$, the Pauli word $P_a$ acts on $\rho$. We generalize this to circuits with multiple layers and insert Pauli words according to a correlated probability distribution.

\begin{definition}[Spatiotemporally correlated Pauli channels]
  Let $U=U_g \ldots U_1$ be a noiseless quantum circuit consisting of the unitaries $\{U_i\}_i$. The spatiotemporally correlated Pauli channels are of the form
  \begin{equation}
    \Lambda(\rho)=\sum_a p_a \tilde U^\nodagger_a \rho\, \tilde U_a^\dagger,
  \end{equation}
  with
  \begin{equation}\label{eq:spatiotemp_pauli_chan}
    \tilde U_a= \prod_{i=g}^1(P_{a_i}U_i)
  \end{equation}
the circuit realization in case of noise realization $a$.
\end{definition}
Spatiotemporally correlated Pauli channels effectively insert Pauli operations into the circuit stochastically and can hence be considered to be a form of stochastic unitary error, where all unitaries are Pauli operators and where the probability distribution may be correlated in space and time. The fluctuator models of the main text are spatiotemporally correlated Pauli channels if the error operation $V$ is a Pauli operator. Note that if a set of transformations generates a symmetry group of the cost function landscape, also the inverses and any composition of those transformations form symmetries of the cost function landscape.

\begin{lemma}\label{thm:symmetries}
Any $n$-qubit QAOA SK cost function landscape $\tilde C(\bgamma,\bbeta)$, obtained using any number of SWAP gates  while the ansatz circuit was acted upon by any spatiotemporally correlated Pauli channel, is invariant under the group generated by the following transformations.
\begin{enumerate}
\item Translation of any $\gamma_k$ by $2\pi$,
\item Translation of any $\beta_k$ by $1\pi$.
\item Negation of any $\beta_k$ (i.e., $\beta_k\mapsto -\beta_k$) and simultaneous translation of $\gamma_k$ and $\gamma_{k+1}$ by $\pi$. (For the last $\beta$ parameter, that is, the edge case $k=r$, only $\gamma_k=\gamma_r$ needs to be translated by $\pi$.)
\item Simultaneous negation of all parameters.
  \end{enumerate}
\end{lemma}

\begin{proof}
Let us first assume the QAOA circuit is implemented in the absence of noise. In the following, we consider two states or two unitaries equal if they are equal up to a global phase. The first generator follows directly from the periodicity of $RZZ(\gamma)$ with period $2\pi$. For the second generator, note that $RX(\beta\pm\pi) =X\,RX(\beta)$.  Thus, sending $\beta_k\mapsto \beta_k\pm\pi$ has the same effect on the ansatz state as adding a layer of $X$ gates immediately before the layer of $RX$ gates of the cycle $k$. Now, note that the layer of $X$ gates commutes with all $RZZ$ gates, any SWAP gates, and all layers of $RX$ gates at all cycles of the ansatz circuit. Thus, we may commute the layer of $X$ gates all the way to the beginning of the circuit and act with it on the initial state, resulting in $X^{\otimes n}\ket +^{\otimes n}=\ket +^{\otimes n}$. Thus, the ansatz state $\ket{\bgamma,\bbeta}$ is invariant under $\beta_k\mapsto \beta_k\pm\pi$ for any $k$, and consequently $\tilde C(\bgamma,\bbeta)$ is invariant under these transformations.

For the third generator, note that $RZZ[\pm'(\gamma\pm''\pi)]=(Z \otimes Z)\,  RZZ[\pm'(\gamma\pm''\pi)]$. (The signs $\pm'$ and $\pm''$ are independent.) Thus, sending $\gamma_k\mapsto \gamma_k\pm'\pi$ has the same effect as adding $Z\otimes Z$ to the circuit immediately after each $RZZ$ gate in the Hamiltonian stage of the cycle $k$. Note that $Z_iZ_j$ commutes with $RZZ_{kl}(x)$ for all $x,i,j,k,l$, and that any SWAP gate can be seen as permuting qubit indices, not changing the number of times a qubit is acted on by an $RZZ$ gate. We may now commute all excess $Z$ gates that arose from the shift $\gamma_k\mapsto \gamma_k\pm''\pi$ to immediately before the layer of $RX$ gates of cycle $k$. At this layer, each qubit is now acted on by an odd number of $Z$ gates, which is equal to the situation where each qubit is acted on by a single $Z$ gate. We can do the same for the cycle $k+1$ of the ansatz circuit, sending  $\gamma_k\mapsto \gamma_k\pm''\pi$, but now moving the resulting $Z$ gates to immediately before the layer of $RX$ gates of the cycle $k$. Using $Z_i\,RX_i(\beta_k) Z_i=RX_i(-\beta_k)$ on all qubits proves the standard case of generator 3. For the edge case $k=r$, sending $\gamma_r\mapsto \gamma_r\pm'' \pi$ results in a layer of $Z$ gates immediately before the last layer of $RX$ gates. Commuting this layer through the last layer of $RX$ gates sends $\beta_r$ to $-\beta_r$. Noting that the layer of $Z$ gates commutes through (any term of) $H$, and thus that $\bra{\bgamma,\bbeta}Z^{\otimes n} H Z^{\otimes n} \ket{\bgamma,\bbeta}=\bra{\bgamma,\bbeta} H \ket{\bgamma,\bbeta}=\tilde C(\bgamma,\bbeta)$, completes the proof of generator 3 (in the noiseless case).

The fourth generator follows because (in the computational basis) ($i$) any gate in the ansatz is generated by a real Hamiltonian, ($ii$) the Hamiltonian used in QAOA is real, and ($iii$), the initial state is real. Both $U(\theta)=RX(\theta)$ and $U(\theta)=RZZ(\theta)$ are generated by real Hamiltonians [i.e., statement ($i$) holds]. Thus, for both operators, $U^{T}(\theta)=U^\dagger(-\theta)$ and, likewise, $U^{\dagger \mrm{T}}(\theta)=U(-\theta)$. From ($ii$) and ($iii$), it follows that $H^\mrm{T}=H$ and $(\ket +^n)^\mrm{T}=\bra{+}^n$. Let the sequence of gates in the ansatz be given by $\Pi_{i=g}^{1}U_i(\theta_i)$, with $g$ the number of gates in the ansatz, and note that trivially $x=x^\mrm{T}$ for $x$ a real number.  Then,
\begin{align}
  \tilde C(\bgamma,\bbeta)&=(\bra +^n [\Pi_{i=1}^{g} U^{\dagger}_i(\theta_i)] H  [\Pi_{i=g}^{1} U_i(\theta_i)] \ket +^n)^\mrm{T}\nn\\
                   &=\bra +^n[\Pi_{i=1}^{g}U^{\dagger}_i(-\theta_i)] H [\Pi_{i=g}^{1}U_i(-\theta_i)] \ket +^n\nn\\
  &=\tilde C(-\bgamma,-\bbeta).
\end{align}
This completes the proof of Lemma~\ref{thm:symmetries} in the noiseless case.

In case the ansatz circuit is acted upon by a spatiotemporally correlated Pauli error channel,
\begin{align}
  \tilde C(\bgamma,\bbeta)=&\sum_a p_a \bra +^{\otimes n} \tilde U_a^\dagger(\bgamma,\bbeta) H\,\tilde U^\nodagger_a(\bgamma,\bbeta) \ket +^{\otimes n}\nn\\
  :=&\sum_a p_a \tilde C_a(\bgamma,\bbeta).
\end{align}
So, to show that $\tilde C(\bgamma,\bbeta)$ is invariant under the transformations generated by generators 1--4, it suffices to show that $\tilde C_a(\bgamma,\bbeta)$, the cost function in case of realization $a$, is invariant under those transformations for all noise realizations $a$.

Transformation 1 trivially leaves $\tilde C_a(\bgamma,\bbeta)$ invariant under any spatiotemporally correlated Pauli error channel. To show that the transformations 2 and 3 leave  $\tilde C_a(\bgamma,\bbeta)$ invariant under the same type of channel, consider $U_a(\bgamma,\bbeta)$, the circuit in case of noise realization $a$. Note that any Pauli operator either commutes or anticommutes with both the layer of $X$ gates arising from any the parameter shift $RX(\beta\pm\pi)=X\,RX(\beta)$, and any $Z\otimes Z$ operators originating from  $RZZ[\pm'(\gamma\pm''\pi)]=Z \otimes Z\  RZZ[\pm'(\gamma\pm''\pi)]$. Hence, we can use the noiseless proofs of generators 2 and 3 to show $C_a(\bgamma,\bbeta)$ is invariant for all $a$ if we make the additional observation that any overall factors of $-1$, arising from anticommutativity, have no effect on the expectation value $C_a(\bgamma,\bbeta)$.

To show that generator 4 leaves $\tilde C_a(\bgamma,\bbeta)$ invariant, note $RZZ_{ij}=RZZ_{ji}$, and that $RZZ(\theta)P_i=P_i\, RZZ_{ij}(-\theta)$, if $P_i=X_i$ or $P_i=Y_i$, and $RZZ_{ij}(\theta)P_i=P_i\, RZZ_{ij}(\theta)$ if $P_i=Z$. Likewise, $RX(\theta) P=P\,RX(-\theta)$ if $P=Y$ or $P=Z$, and $RX(\theta) P=P\,RX(\theta)$ if $P=X$. Commuting all Pauli operators arising from an error realization to the beginning of the circuit, we obtain $\tilde U_a(\bgamma,\bbeta)= U_{\boldsymbol{0}}(\tilde \bgamma,\tilde \bbeta)P$, with $P$ some tensor product of Pauli operators, $\tilde U_{\boldsymbol 0}$ the noiseless circuit realization, and where $\tilde \bgamma,\tilde \bbeta$ can be determined explicitly using the aforementioned commutation relations. Thus, $U_a(\bgamma,\bbeta)\ket +^{\otimes n}=  U_{\boldsymbol{0}}(\tilde \bgamma,\tilde \bbeta) \ket \psi$, with $\ket \psi$ a state with real entries. Hence, conditions $i$-$iii$ used in the noiseless proof of generator 4 are satisfied, and therefore $C_a(\bgamma,\bbeta)=C_a(-\bgamma,-\bbeta)$.
\end{proof}

\subsection{Parameter representatives}\label{sec:parameter_representatives}
Due to the symmetries of the SK QAOA cost function landscape, two unequal sets of parameters $\Phi=(\bgamma,\bbeta)$ and $\Phi'=(\bgamma',\bbeta')$ may be deemed equivalent if there exists a group element from the symmetry group of the QAOA SK cost function landscape that maps $\Phi$ to $\Phi'$. This defines equivalence classes of sets of parameters. Algorithm~\ref{alg:rep} maps each set of parameters $\Phi$ to a unique representative $\mc A(\Phi)$ of the equivalence class that set of parameters is in. Thus, the Euclidean distance $D$ between $\Phi$ and $\Phi'$, modulo the symmetries of the cost function, may be defined as
\begin{equation}\label{eq:D}
D=\sqrt{\sum_i\left[\mathcal A(\Phi)_i-\mathcal A(\Phi')_i\right]^2}.
\end{equation}
This is the procedure used in the Supplemental Material~\cite{Note3} to plot the distance $D$ between the noise-unaware parameters and the noise-aware parameters.

\begin{algorithm}
  \hrule
  \includegraphics{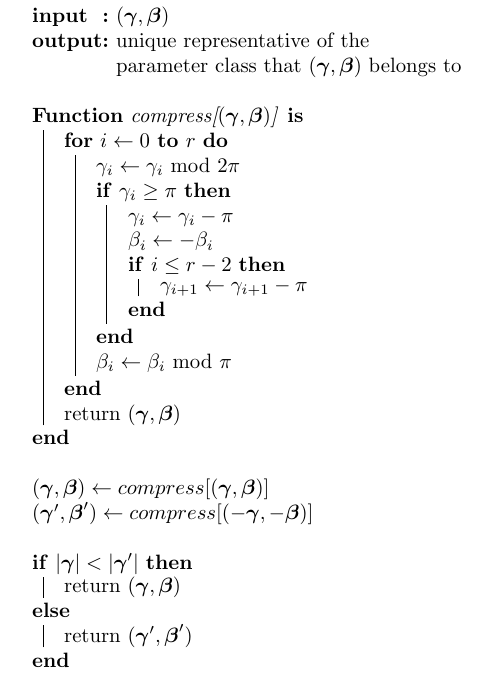}
  \hrule
  \caption{\label{alg:rep} The algorithm that removes the redundant freedom that is caused by symmetries of the QAOA SK cost function landscape. It does so by mapping each set of parameters to a unique class representative.}
\end{algorithm}

\section{Instance properties}\label{sec:instance_properties}
The instance reported on in \sec{numerical_methods_and_results} is
\begin{equation}
H=\sum_{i<j}\om_{ij}Z_iZ_j, 
\end{equation}
with 
\begin{equation}
\om=
    \left( 
        \begin{array}{cccccc}
            0 & +1 & -1 & +1 & +1 & -1 \\
            0 & 0 & +1 & -1 & +1 & +1 \\
            0 & 0 & 0 & -1 & -1 & -1 \\
            0 & 0 & 0 & 0 & -1 & -1 \\
            0 & 0 & 0 & 0 & 0 & +1 \\
            0 & 0 & 0 & 0 & 0 & 0
        \end{array}
    \right).
\end{equation}
We can write the data defining the instance more compactly by gathering and concatenating the factors $\pm 1$ from left to right and top to bottom. If we additionally make the identification  $(+1,-1)\mapsto (0,1)$, we obtain the \emph{instance bit string} 
\begin{equation}\label{eq:instance_bit_string}
\tilde \om = 010010100111110.
\end{equation}
In general, the instance bit string of a SK model on $n$ qubits contains $n_{\tilde{\om}}=n(n-1)/2$ bits and reads $\tilde \om=\prod_{i=1}^{n-1}\prod_{j=i+1}^n[(1-\om_{ij})/2]$. Here, the product acts by concatenation (not multiplication) and the bitstring is built up from left to right (i.e., if $n=6$, $\tilde \om_1=(1-\om_{1,2})/2$, $\tilde \om_2=(1-\om_{1,3})/2$, \ldots, $\tilde \om_{15}=(1-\om_{5,6})/2$).

We sampled 16 random SK problem instances (including the one above) on $n=6$ qubits by sampling 16 bit strings of length $15$ uniformly at random. Table~\ref{tab:instances} lists the following properties of these instances.
\begin{table}
  \centering
  \begin{tabular}{l|l|l|l|l}
    \parbox{8em}{\phantom{line} \\ \raggedright Instance \\ bit string\\ \vspace{.5em}} & \parbox{3em}{\raggedright Graph\\\raggedright class\\ card.\\  \vspace{.5em} } & \parbox{7em}{\phantom{line} \\\phantom{line}\\\raggedright GS manifold \\ \vspace{.5em}} & \parbox{1em}{\phantom{line}\\\phantom{line}\\\raggedright  $E_{\mrm{GS}}$ \\ \vspace{.5em}}& \parbox{3.3em}{\phantom{line}\\ \raggedright $\tilde{C}(\bgamma,\bbeta)$ \\ \raggedright class \\ \vspace{.5em} }\\\hline \hline
000100100011110\ 	&	360	&	010101	&	-7	&	0 	\\	\hline
001100110101110	&	180	&	001001	&	-9	&	1\\
001111110000010	&	720	&	011010	&	-9	&	1\\
010011011101010	&	180	&	010010	&	-9	&	1\\\hline
001111011010001	&	720	&	011011	&	-9	&	2\\
010100101011100	&	360	&	010101	&	-9	&	2\\
100010101000100	&	360	&	001110	&	-9	&	2\\
110010010101000	&	720	&	010010	&	-9	&	2\\\hline
010010100111110*	&	720	&	010000, 010110 & -7	&	3\\
011100011010011	&	360	&	010011, 010101	&	-7	&	3\\
100110010001101	&	720	&	001001, 001100	&	-7	&	3\\
101100001010110	&	360	&	001010, 010001	&	-7	&	3\\
111011011000010*	&	720	&	000010, 011010	&	-7	&	3\\
111100111011001	&	360	&	000100, 001011	&	-7	&	3\\\hline
011101001011100	&	180	&	011001	&	-11	&	4\\\hline
\parbox{1em}{\vspace{.4em} \raggedright 110101111100011\\ \ \\ \ }	&\parbox{1em}{\vspace{.4em}\raggedright 72\\ \ \\ \ }	&	\parbox{7em}{\vspace{.4em}\raggedright 000000, 000011,\\\raggedright 000101, 001100,\\\raggedright 010111, 011101}	&\parbox{2em}{\vspace{.4em}\raggedright -5 \\ \ \\ \ }	&	\parbox{0em}{\vspace{.4em}\raggedright 5\\ \ \\ \ }
  \end{tabular}
  \caption{\label{tab:instances}The random instances and some of their properties. None of the 16 graphs are mutually isomorphic, except for those noted with an asterisk.}
\end{table}
\begin{enumerate}
  \item \emph{Instance bit string.}  
  \item \emph{Graph class cardinality.} Many of the $2^{n_{\tilde \om}}$ SK Hamiltonians on $n$ qubits are equivalent up to a relabeling of the qubit indices. That is, many of the SK graphs with $n$ nodes, obtained by interpreting the $\om_{ij}$ as entries of an adjacency matrix, are isomorphic. Indeed, for $n=6$, we found by a brute-force method that the $2^{n_{\tilde \om}}=32\,768$ SK graphs fall into 156 different equivalence classes, where equivalence is defined by graph isomorphism. The graph class cardinality is the number of distinct graphs in a graph class. The highest occurring class cardinality is 720, and there are 8 classes attaining this cardinality. From the 16 random graphs we generated, only the graphs with instance bit strings 010010100111110 and 111011011000010 are in the same class (i.e., isomorphic). The class they belong to has the highest occurring cardinality of 720.

  \item  \emph{Ground state manifold.} We indicate the ground state manifold of the SK instances by lists of bit strings. The list of an instance is related to the instance's ground state manifold GS by $\mrm{GS}=\mrm{span}(\cup_i\{\ket i,\ket{\bar i}\})$, where $i$ ranges over all bit strings in the list and $\bar i$ indicates the negation of $i$. It follows from $[X^{\otimes n},H]=0$ that $\ket{\bar i}$ is in the ground state manifold if $\ket{i}$ is in the ground state manifold.

  \item \emph{Ground state energy.} The energy $E_{\mrm{GS}}$ of any state in the ground state manifold.

  \item \emph{Cost function landscape class.} Consistent with QAOA's cost function landscape concentration~\cite{brandao2018fixed}, we find that two nonisomorphic SK instances may have identical cost function landscapes (up to numerical precision) at $n=6,d=3,p=0$. We consider two SK instances to be in the same cost function landscape class if their cost function landscapes are identical. The 16 random SK instances studied in this work, of which only two are isomorphic, fall into 6 distinct cost function landscape classes.
\end{enumerate}

Cost function landscape equivalence is established numerically by the comparison of cost function landscape values at 64 random points in parameter space. We deem a pair of SK instances equivalent to numerical accuracy if their cost function landscape values differ by less than $10^{-14}$ in absolute value across all 64 random points. We found that, by this criterion, the set of pairs of instances that lie in the same cost function landscape class is well separated from the set of pairs of instances that do not. (For each pair of instances in the same cost function landscape class, there is no point in parameter space among the 64 sampled points where the pair's cost functions differ by more than $10^{-4}$. For each pair of instances \emph{not} in the same cost function landscape class, we found that there is at least one point in parameter space where their cost functions differ by more than unity.) Increasing the error probability $p$ breaks the cost function equivalence. The data presented in the Supplemental Material~\cite{Note3} show that this does not lead to qualitative differences in the AR of instances in the same cost function landscape class.

In the main text, we have shown the data pertaining to the SK instance with instance bit string 010010100111110. We consider this instance typical because it falls into a graph class with the highest possible cardinality and because the qualitative performance of QAOA on this instance is largely equal to the qualitative performance of QAOA on all other sampled instances.

\end{document}